\documentclass[conference]{IEEEtran}
\usepackage{amssymb}
\usepackage{amsfonts}

\usepackage{amsmath,amssymb}%
\ifCLASSINFOpdf 
\usepackage[pdftex]{graphicx} 
\usepackage[pdftex,colorlinks, 
           bookmarks=true,%
           pdftitle=SINR-based\ k-coverage\ probability\ in\ cellular\ networks\ with\ arbitrary shadowing,%
           pdfauthor=B.\ Blaszczyszyn%
           \ -\  M.\ Karray\ -\ P.\ Keeler,]{hyperref}  
\usepackage[numbers,sort&compress]{natbib} 
\usepackage{hypernat} 
\hypersetup{
   bookmarksnumbered,
   pdfstartview={FitH},
   citecolor={blue},
   linkcolor={red},
   urlcolor=[rgb]{0,0.55,0},
   pdfpagemode={UseOutlines}}
\else

 \usepackage[dvips]{graphicx} 
 \usepackage[numbers,sort&compress]{natbib} 
\fi

\newtheorem{thm}{Theorem}
\newenvironment{theorem}{\bf\begin{thm}\rm\em}{\end{thm}} 
\newtheorem{cor}[thm]{Corollary}
\newenvironment{corollary}{\bf\begin{cor}\rm\em}{\end{cor}} 
\newtheorem{lem}[thm]{Lemma}
\newenvironment{lemma}{\bf\begin{lem}\rm\em}{\end{lem}} 
\newtheorem{prop}[thm]{Proposition}
\newtheorem{rem}[thm]{Remark}
\newenvironment{remark}{\bf\begin{rem}\rm}{\end{rem}} 

\newcommand{\E}{\mathbf{E}}

\newcommand{\calI}{\mathcal{I}}
\newcommand{\calJ}{\mathcal{J}}
\newcommand{\calS}{\mathcal{S}}

\newcommand{\calL}{\mathcal{L}}

\newcommand{\Prob}{\mathbf{P}}

\newcommand{\Ind}{\mathbf{1}}
\newcommand{\R}{\mathbb{R}}

\newcommand{\calN}{\mathcal{N}}

\newcommand{\Lap}{\mathcal{L}}
\newcommand{\SINR}{\text{SINR}}

\title{\Huge  SINR-based $k$-coverage probability in cellular networks with arbitrary shadowing}

\author{{\bf H.P. Keeler}$^\dagger$, {\bf B. B{\l}aszczyszyn}$^\dagger$
and {\bf    M. K. Karray}$^\ast$}
\usepackage[scriptsize,it]{caption}   
\begin{document}

\maketitle

\begin{abstract}
We give numerically tractable, explicit integral expressions for the distribution of the signal-to-interference-and-noise-ratio (SINR) experienced by a typical user in the down-link channel from the $k$-th strongest base stations of a cellular network modelled by Poisson point process on the plane.  Our signal propagation-loss model comprises of a power-law path-loss function with arbitrarily distributed shadowing, independent across all base stations, with and without Rayleigh fading.  Our results are valid in the whole domain of SINR, in particular for $\SINR<1$, where one observes multiple coverage.  In this latter aspect our paper complements previous studies reported in~\cite{DHILLON2012}. 
\end{abstract}


\section{Introduction}
\let\thefootnote\relax\footnotetext{\hspace{-2ex}$^\dagger$Inria/Ens, 23
  av. d'Italie 75214 Paris, France\\ 
$^\ast$Orange Labs, 38/40 rue
G\'{e}n\'{e}ral Leclerc, 92794 Issy-Moulineaux, France}
\newcommand{\thefootnote}{\arabic{footnote}}
Shannon's theory and its modern extensions 
quantify  the quality of communications channels (ergodic capacity,
finite block errors, error exponents, etc.) in a probabilistic manner by
considering  averages over codewords and  channel characteristics
(noise, fading, etc.). For communication networks with  many channels,
it has been recently suggested to use a stochastic geometric
approach~\cite{HABDF:2009} consisting in taking 
spatial averages over node (emitter, receiver) locations.
Establishing clear connections between stochastic-geometric averages and basic
information-theoretic notions may be difficult 
(cf  e.g.~\cite{Piantanida2011}) but
this approach for wireless networks
has recently attracted a lot of attention. In particular, the
fundamental characteristic, discovered in information theory,  signal-to-noise-ratio
is now being studied in many geometric contexts with the incorporation of interference. This paper contributes to this approach by considering the distribution of 
the signal-to-interference-and-noise-ratio
(SINR) of a typical user on down-link channels from different base
stations of a single-tier cellular network modelled by Poisson point process on
the plane. In particular, it complements~\cite{DHILLON2012}, by
providing explicit characterization of low (less than one) values of
SINR.  Current cellular-network technology allows for 
effective use of such SINR regimes, whence our motivation
comes.

Cellular network models based on the Poisson point
process have been shown to give tractable and accurate
solutions~\cite{ANDREWS2011}, with the Poisson assumption
being  justified  by representing highly irregular base station deployments in urban
areas~\cite{ChiaHan_etal2012}  or mimicking strong log-normal
shadowing~\cite{hextopoi}, or  both.  Knowledge of the distribution of
$\mathrm{SINR}$ allows  to calculate  key performance
indicators of cellular networks, e.g. spectral 
efficiency~\cite{AlouiniGoldsmith1999}~\footnote{From~\cite[\S
  4.2.2]{BonaldBorstHegdeJP2009} 
we know that it represents the  critical
traffic demand per base station, beyond which the best-effort service
of variable-bit-rate traffic becomes unstable.}
or energy efficiency~\cite{Richter2009,hextopoi}. It can also be
(non-trivially) related to user-level quality-of-service metrics.

In studying the SINR of a given user 
with respect to the base station with the greatest received signal, the main difficulty for small values of SINR is taking
into account {\em multiple coverage}. Simple algebra shows
that there can be at most one base station offering a given user
$\SINR\geq1$, and hence the probability of having at least one station
covering at this level reduces to the sum of probabilities of
SINR-coverage over all base stations. However, this is not the case for
$\SINR<1$, where  one  needs to 
study probabilities of simultaneous coverage by several base stations.  
We express these probabilities via the so-called {\em symmetric sums}
and relate them to the appropriate partitioning of the SINR domain for
$\text{SINR}<1$, which are the main  ideas behind this paper.

Our  signal propagation-loss  model consists of the deterministic power-law
path-loss function and independent  (across all base stations), arbitrarily distributed shadowing.
As  previously observed~\cite{shadow_jnl,hextopoi}, 
any characteristic involving only the sequence of propagation-loss
values experienced  by a given user from all base stations entirely depend on
the distribution of the shadowing via its moment of order $2/\beta$,
where $\beta$ is the path-loss exponent.
We also study the impact of Rayleigh fading on the SINR coverage by
incorporating into the propagation-loss model additional independent (across all base stations)
random variables with exponential distributions. Assuming that fading
affects the  SINR coverage condition but not the choice of the serving 
base station, we observe that it worsens the coverage
particularly at small values of SINR.


\subsection*{Related work}
The SINR coverage in a multi-tier network
was studied in~\cite{DHILLON2012} for $\SINR\ge 1$.
Two different approaches to express the distribution of the SINR in
its whole domain, both involving inversion of Laplace transforms, were 
presented in~\cite{DHILLON2011,hextopoi}. 
Our expressions, which involve two key families of integrals 
(over the positive real line and a hyper-cube respectively) 
are much more tractable  in numerical evaluation.  
Our model with fading  was recently examined in~\cite{VU2012} under slightly
more general assumptions. We revisit it in order to 
present more closed-form expressions under our specific assumptions,
and compare the coverage probabilities obtained in it to these 
in the model  without fading.~\footnote{%
Which of these two models is more appropriate depends on the context.
For example if users are motionless, then the fading should be taken
into account, while highly mobile users ``see'' channel
characteristics averaged over fading, cf~\cite[Proposition~3.1]{Borst2003}.  
}

\section{Model description}
\label{s.Model}
On $\R^2$, we model the base stations with a homogeneous Poisson point process
 $\Phi$ with density $\lambda$. Given $\Phi$, 
let $\{S_x\}_{x\in\Phi}$ be a collection of
independent and identically distributed random
variables,  that 
represent the {\em shadowing} experienced between station $x\in\Phi$
and a typical user located, without loss of generality,
at the origin. Let $S$ denote a generic
shadowing variable. The distribution of $S$ is arbitrary
except for a technical assumption  $\E[S^{\frac{2}{\beta}}]<\infty$  
and a  conventional assumption  that $\E[S]=1$ that we make without
loss of generality.
\subsection{SINR multi-coverage}
\label{ss.SINRcoverage} 
We define the SINR of the typical
user with respect to the station $x\in\Phi$ by
\begin{equation}\label{SINR}
\SINR(x):=\frac{S_x/\ell(|x|)}{W+I-S_x/\ell(|x|)},
\end{equation}
where the constant $W$ is the noise power,
$I=\sum_{x\in\Phi}S_x/\ell(|x|)$ is the total power received from  the
entire  network, and the path-loss function is
\begin{equation}\label{PATHGAIN}
\ell(|x|)=(K|x|)^{\beta},
\end{equation}
with constants $K>0$ and $\beta>2$. 
In this paper we are interested in  the distribution of the 
\emph{coverage number} of the typical user defined as the
number of base stations that the typical user can connect to at the
SINR level $T$, namely
\begin{equation}
\calN(T)=\sum_{x\in\Phi}\Ind \left[\SINR(x)>T\right].
\end{equation}
The probability  of the typical user being covered by
at least $k$ base stations, which we call \emph{$k$-coverage probability}, is
\begin{equation}
P_c^{(k)}(T)
=\Prob\{\,\calN(T)\geq k\,\}.
\end{equation}
In particular, the {\em coverage probability} of the typical
user  is (\footnote{This notation is similar to that of~\cite{DHILLON2012},
which uses $\beta$ instead of $T$.})
$P_c(T):=P_c^{(1)}(T)$.

Since the function $x/(A-x)=A/(A-x)-1$ is increasing in $x$,
$P_c^{(k)}(T)$, as a function of $T$, is the 
tail-distribution function of the SINR experienced by the typical
user with  respect to the base station offering the 
 $k\,$th smallest
propagation-loss $Y_{k}$: $(Y_k)^{-1}/(W+I-(Y_k)^{-1})$, where
$Y_1<Y_2<\ldots$ is the process of order statistics  
of $\{\ell(|x|)/S_x: x\in\Phi\}$. In particular, 
$P_c(T)$ is the tail-distribution function of the SINR with respect to the 
base station with the smallest propagation-loss.

Related quantities of interest include also 
the {\em expected converge number}  $\E[\calN(T)]=
\sum_{n=0}^{\infty} n \Prob\{\,\calN(T)=n\,\}$
and its probability-generating function $G(z)=\E[z^{\calN(T)}]$.

\subsection{Adding fading to the model}
\label{ss.Fading_model}
In this extension of the previous 
model we assume that the propagation-loss
of each base station $x\in\Phi$ is
further modified by a random {\em fading} variable $F_x$ 
and equal to  $l(|x|)/(S_xF_x)$,
where given
$\Phi$,  $\{F_x\}_{x\in\Phi}$ is a collection 
of independent and identically distributed random
variables, independent of shadowing $\{S_x\}_{x\in\Phi}$. In this
paper we will assume Rayleigh fading, i.e., that the generic fading
variable $F$ is exponential, with $\E[F_x]=1$. 
A key assumption is that {\em fading perturbs the
SINR coverage condition but not the choice of the serving base
station}~\footnote{In other words, the user  compares the received signals 
averaged over fading effects, which is justified by 
short time and space coherence properties of the (multipath) fading.}.
In consequence, the {\em coverage probability under fading}
(with respect to the smallest fading-averaged-propagation-loss base station)
 is defined as 
\begin{equation}\label{e.Pc_fading}
\tilde P_c(T)=
\Prob\left\{\,\frac{(Y_1)^{-1}F}{W+I-(Y_1)^{-1}F}>T\,\right\}\,,
\end{equation}
where, recall,  $Y_1$ is the smallest propagation-loss received by the
typical user in the model without fading.

\section{Preliminary Observations}
\subsection{Invariance with respect to the shadowing distribution}
\begin{lemma}[Cf~\cite{shadow_jnl,hextopoi}]
\label{l.Yn}
The fading-averaged-propagation-loss process $(Y_n:n\ge
1)$, considered as a point process on the positive half-line $\mathbb{R}^{+}$
is a non-homogeneous Poisson point
process with intensity measure
$\Lambda\left(  \left[  0,t\right)  \right)  =at^{\frac{2}{\beta}}$
where
\begin{equation}\label{e.a}
a:=\frac{\lambda\pi
  \E[S^{\frac{2}{\beta}}]}{K^{2}}\,.
\end{equation}
\end{lemma}
Consequently, 
the distribution of $(Y_n:n\ge1)$, and hence the functions
$P_c^{(k)}(T)$ and   $\tilde P_c(T)$, depend on the model parameters
(including the shadowing $S$ distribution) only though the noise level
$W$, path-loss exponent $\beta$ and the constant $a$.\footnote{%
\label{foot:invariance}This means that evaluating our quantities  of interest in our model with a
general distribution of shadowing $S$ and some value of the
constant $K$, one can equivalently and for mathematical convenience 
assume  some particular distribution $\tilde S$
of shadowing, e.g. exponential or constant, and $\tilde K=1$ provided 
one replaces $\lambda$ by $\lambda (\tilde
K/K)^2\E[S^{\frac{2}{\beta}}]/\E[\tilde S^{\frac{2}{\beta}}]$ in the
obtained formula. 
We will use these two representations of the general 
models when evaluating $\tilde
P_c^{(k)}(T)$ and   $P_c(T)$, respectively, in Sections~\ref{ss.Results} and~\ref{ss.Results_fading}.}

\subsection{Symmetric sum representation}
For any given $T$  and $n\ge 1$ define the  $n\,$th  symmetric sum  
\begin{equation}\label{e.SS}
\calS_n(T):=\E\Bigl[\sum_{{x_1,\ldots,x_n\in\Phi\atop\text{distinct}}}
\Prob\{\,\SINR(x_i)>T,\,i=1,\ldots,n\,|\,\Phi\,\}\Bigr]\,,
\end{equation}
where $\Prob\{\,...\,|\,\Phi\,\}$ denotes the  conditional probability given
$\Phi$ (with random shadowing marks).
We set $\calS_0(T)\equiv 1$, and note that $\calS_n(T)$ is the expected number of ways that the typical user can connect to $n$ base stations  when there are $\calN(T)$ base stations each with a SINR greater than $T$.
 We have the following 
identities related to the famous 
inclusion-exclusion principle (cf e.g~\cite[IV.5 and
IV.3]{Fel68} for~(\ref{e.ss1}) and~(\ref{e.ss2}), respectively).%
\footnote{A general relation between the distribution of $\calN$
and the symmetric sums is given by the Schuette-Nesbitt formula,
often used in insurance mathematics, cf~\cite{GERBER:1995}.}
\begin{lemma}\label{l.SS-represent}
We have for  $k\ge1$ 
\begin{eqnarray}
 P_c^{(k)}(T)&=&\sum_{n=k}^{\infty}  (-1)^{n-k}{n-1\choose k-1}\calS_n(T)\,,\label{e.ss1}\\
\Prob\{\,\calN(T)=k\,\}&=&\sum_{n=k}^{\infty}  (-1)^{n-k}{n\choose k}\calS_n(T)\,,\label{e.ss2}\\
 \E[z^{\calN(T)}]&=&\sum\limits_{n=0}^{\infty} (z-1)^n \calS_n(T)\,,
\quad z\in[0,1]\,, \label{e.ss3}\\
 \E[\calN(T)]&=&\calS_1(T)\,.\label{e.ss4}
\end{eqnarray}
\end{lemma}
Our goal in Section~\ref{ss.Results} will be to evaluate symmetric sums
$\calS_n(T)$, which will allow us to express easily our  quantities of
interest appearing in the right-hand side of the above expressions.
Before doing this, in the following section we explain that 
the (apparently infinite) summations presented above boil down to finite
sums, as for any given $T$ we have  $\calS_n(T)=0$ for $n$ large enough.

\subsection{Partition of the $T$-domain}
For real $x$ denote by $\lceil x\rceil$  the ceiling of $x$ (the smallest integer not less than $x$).
\begin{lemma}\label{l.Tdomain}
For $n\ge1$, $\calS_n(T)=0$ whenever  $T\ge 1/(n-1)$.
In other words, one can replace
$\infty$ by $\lceil 1/T \rceil$  in the sums in
expressions given in Lemma~\ref{l.SS-represent}.
\end{lemma}

\begin{proof}
This is stems from a well-known constraint of the SINR cell
intersection. If the SINR of a given user with respect to $n$ distinct
stations is to be larger than $T$, then $nT/(1+T)\le1$ (cf
\cite[Proposition~6.2]{FnT1}, with the strict inequality holding whenever
there are other (interfering) stations or external
noise, which is the case in our model. Hence,
for $T\ge 1/(n-1)$ all the terms (probabilities) in the $m$th symmetric
sum $\calS_m(T)$ in~(\ref{e.SS}) are null for every $m\ge n$.
\end{proof}
\section{Main results}
\subsection{Key integrals}
We now introduce two families of 
functions  which will allow us to express $\calS_n(T)$
and, in consequence, the multi-coverage characteristics in the model without fading.  
For $x\ge0$ define
\begin{equation}\label{In}
\calI_{n,\beta}(x)=\frac{2^n
\int_0^{\infty} u^{2n-1}e^{-u^2-u^\beta x\Gamma(1-2/\beta)^{-\beta/2}} du
}{\beta^{n-1}(C'(\beta))^n(n-1)!}
\end{equation}
where 
\begin{equation}
 C'(\beta)=\frac{2\pi}{\beta\sin(2\pi/\beta)}=
\Gamma(1-2/\beta)\Gamma(1+2/\beta).
\end{equation}

\begin{remark} 
We have 
\begin{equation}\label{InW0}
\calI_{n,\beta}(0)=  \frac{2^{n-1}}{\beta^{n-1}(C'(\beta))^n}\,.
\end{equation}
\end{remark}

The second family of functions are integrals over the hyper-cube. 
For $x\ge0$ define
\begin{equation}\label{e.Jn}
 \calJ_{n,\beta}(x)=\int_{[0,1]^{n-1}}
 \frac{   \prod\limits_{i=1}^{n-1}   v_i^{i(2/\beta+1)-1}(1-v_i)^{2/\beta}
  }{\prod\limits_{i=1}^{n-1} (x+\eta_i)} dv_1\dots
dv_{n-1}
\end{equation}
where 
$\eta_i:=(1-v_i)\prod_{k=i+1}^{n-1}v_k$.

\begin{remark}\label{r.hypergeometric}
For $\calJ_2$, a closed-form solution exists
\begin{align}
&\calJ_{2,\beta}(x)\\
&=\frac{B(2/\beta+1,2/\beta+1){}_2F_1(1,2/\beta+1;2(2/\beta+1);-1/x
)}{x},\nonumber
\end{align}
where   $B$ is the beta
function~\cite[eq. 5.12.1]{DLMF} and ${}_2F_1(a,b;c;z)$ is the
hypergeometric function given by~\cite[eq. 15.11]{DLMF} (whose
integral representation follows from eq. 15.1.2 and
15.6.1 therein).%
~\footnote{We note that the form of $\calJ_{n,\beta}$ is similar to
integral representations of the generalized hypergeometric function
and a related integral generalization
\cite{ULANSKII:2006}. A closed-form solution of
  $\calJ_{n,\beta}(x)$ may exist, but that is left as a future
  task. For low and intermediate $n$, regular numerical and Monte
  Carlo methods work well and  give results in a matter of seconds on
  a standard PC machine; cf~\cite{paul_matlab}.
 For high $n$, analysis of the kernel of $\calJ_n$ may lead to judiciously choosing suitable lattice rules, thus allowing for relatively fast integration \cite{kuo2005lifting}.}
\end{remark} 

\subsection{Results for the model without fading}
\label{ss.Results}
For $0<T<1/(n-1)$ define
\begin{equation}\label{Tn}
 T_n=\frac{T}{1-(n-1)T}.
\end{equation}

We now present the key result for the model without fading, which gives an explicit expression for the symmetric sums $\calS_n(T)$.
\begin{theorem}\label{mainResult}
Assume shadowing moment condition $\E(S^{2/\beta})<\infty$. 
Then 
\begin{equation}
\calS_n(T)=T_n^{-2n/\beta}  \calI_{n,\beta}(Wa^{-\beta/2})  \calJ_{n,\beta}(T_n)
\end{equation}
for $0<T<1/(n-1)$  and $\calS_n(T)=0$ otherwise, 
where $a$ is  given by~(\ref{e.a}) and $T_n$ by~(\ref{Tn}). 
\end{theorem}

Theorem~\ref{mainResult} in conjunction of Lemma~\ref{l.SS-represent} 
and Lemma~\ref{l.Tdomain} give us in particular the following
expression for the $k$-coverage probability:
\begin{corollary}
Under the assumptions of Theorem~\ref{mainResult}
$$P_c^{(k)}(T)\\
=\sum_{n=k}^{\lceil 1/T\rceil}  \scriptstyle{(-1)^{n-k}{n-1\choose  k-1}}
T_n^{-2n/\beta}  \calI_{n,\beta}(Wa^{-\beta/2})  \calJ_{n,\beta}(T_n)\,,
$$
\end{corollary}
The special case $k=1$, for $T\geq1$  reduces this to
expression~\cite[eq. (25)]{hextopoi}, which is in turn a special case
of~\cite[eq. (2)]{DHILLON2012} for a  single-tier network.

\begin{proof}[Proof of Theorem~\ref{mainResult}]
Assume $0<T<1/(n-1)$ (otherwise the result follows from
Lemma~\ref{l.Tdomain}).  
Following the remark in Footnote~\ref{foot:invariance} we
will first evaluate $\calS_n(T)$ assuming exponential distribution of $S$
and $K=1$, and then bring back the general assumptions appropriately
rescaling the Poisson intensity $\lambda$.
By the (higher order)  Campbell's formula and the Slivnyak's theorem 
(see \cite[(9.10) and~(9.16)]{FnT1}) and a simple algebraic manipulation
\begin{align}\label{e.Expression1}
\calS_n(T)&=
\frac{(2\pi\lambda)^n}{n!}\!\!\!\!\int\limits_{(\R^+)^n}\!\!\!\!
\Prob\Bigl\{\,
\bigcap_{i=1}^{n}(\SINR'(r_i)>T')\,\Bigr\}r_1dr_1\dots r_ndr_n 
\end{align}
where $\SINR'(r_i):=\frac{S_i/\ell(r_i)}{(W+I+\sum_{j=1}^n
S_i/\ell(r_i))}$
and $T':=T/(1+T)$
with $I$ is as in~(\ref{SINR}) and $S_i$ 
exponential (mean-1) variables, mutually independent and independent
of~$I$. 
Moreover, the event whose probability is calculated in~(\ref{e.Expression1})
is equivalent to 
\begin{align}\label{probctprime}
\Bigl\{\,
\frac{\min\left(S_1/\ell(r_1),\dots,S_n/\ell(r_n)\right)}{W+I+\sum_{j=1}
^nS_j/\ell(r_j)}
>T'\,\Bigr\}\,.
\end{align}

For integer $i\in[1,n]$  denote $E_i:=S_i/\ell(r_i)$.
By our previous assumption $E_i$ are independent exponential variables 
with means $1/\mu_i=1/\ell(r_i)$, respectively.
Let $E_{M}:=\min(E_1,E_2,\dots,E_n)$. Note that  $E_{M}$
is exponential variable with mean $1/\mu_M=1/(\sum_{i=1}^n \mu_i) $.
Moreover, define the random variable
$D:=\sum_{i=1}^n E_i -nE_{M}$. By the 
 memory-less property of the exponential distribution note 
the random variable $D$ is independent of $E_M$ and
has  a mixed exponential distribution
characterized by its Laplace transform
\begin{equation}\label{e.LTD}
 \Lap_{{D}}(\xi)
=\frac{\prod_{i=1}^{n}\mu_i}{\sum_{i=1}^{n}\mu_i} 
\,\sum_{i=1}^n \frac{1}{\prod_{j=1,j\neq i}^n (\mu_j+\xi
)}.
\end{equation}
Using the new random variables we can express the
 event~(\ref{probctprime}) as 
$\{\,
E_M>T_n(W+I+D)\,\}$
where $T_n$ is given by~(\ref{Tn}).
Consequently, the probability $\Prob\{\ldots\}$ 
calculated in~(\ref{e.Expression1})
is equal to 
$\Lap_{W}(\mu_{M}T_n)\Lap_I(\mu_{M}T_n)\Lap_{{D}}(\mu_{M}T_n)$,
which is a product of three Laplace transforms. The first transform  is simply
$\Lap_{W}(\xi)=e^{-W\xi}$, the second 
can be shown (\cite[equation 2.25]{FnT1}) to be 
$\Lap_I(\xi)=e^{-\lambda\xi^{2/\beta}\pi C'(\beta)/K^2}$
while the last one is given in~(\ref{e.LTD}).
After substituting the explicit path-gain function (\ref{PATHGAIN}),
noting that there is some symmetry in the integration variables $r_i$,
changing the integration variables $s_i:=r_i(\lambda T_n^{2/\beta}\pi
C'(\beta))^{1/2}$ and  replacing   $\lambda$ by $a/(\pi \Gamma(1+2/\beta))$ to
revoke the exponential shadowing assumption and bring back  the
general distribution of shadowing and constant $K$ (cf
Footnote~\ref{foot:invariance}) one obtains
\begin{align}
&\calS_n(T)=\frac{2^n}{T_n^{2n/\beta}(C'(\beta))^n(n-1)!}
\int_0^{\infty}\!\!\!\!\!\!\!\dots\!\!\!\int_0^{\infty}\\
&\frac{e^{-(\sum\limits_{i=1}^{n} s_i^\beta)^{2/\beta}}\left(\prod\limits_{i=1}^{n} s_i^{\beta+1} e^{-W(a\Gamma(1-2/\beta))^{-\beta/2}s_i^{\beta}}
\right) }{\left(\sum\limits_{i=1}^{n} s_i^{\beta}
\right)\left(\prod\limits_{i=2}^n [s_i^{\beta}+T_n\sum\limits_{k=1}^{n}
s_k^{\beta}
] \right)} ds_1\dots ds_n. \label{intgralHn1}
\end{align}

A substitution of $n$-dimensional spherical-like variables (detailed
in the appendix, Section \ref{mainResultProofPart2})
completes the proof.
\end{proof}

\subsection{Effects of Rayleigh fading}
\label{ss.Results_fading}
We now consider the model with fading; cf
Section~\ref{ss.Fading_model}.
\begin{theorem}\label{t.with-fading}
The coverage probability under fading (defined in~(\ref{e.Pc_fading})
is equal to 
\begin{align}\label{e.tPc}
\tilde P_c(T)
&=\frac{2}{\beta}\int_0^{\infty}  t^{\frac{2}{\beta}-1}  e^{-tTWa^{-\beta/2}}   e^{- t^{\frac{2}{\beta}} }    \\
&\times \exp\left(-\frac{2}{\beta} \frac{  T t^{2/\beta }
{}_2F_1(1,1-2/\beta;2-2/\beta;-T )}{(1-2/\beta) 
} \right)  dt\,, \nonumber
\end{align}
where, again, ${}_2F_1(a,b;c;z)$ is the hypergeometric function
mentioned in Remark~\ref{r.hypergeometric} above.   
\end{theorem}
\begin{remark}
The expression~(\ref{e.tPc})  can be easily evaluated
numerically. 
Setting $W=0$ yields an analytic solution
\begin{equation}
\tilde P_c(T)=\left[1+ \frac{2}{\beta} \frac{  {}_2F_1(1,1-2/\beta;2-2/\beta;-T )}{(1-2/\beta) }T\right]^{-1}.
\end{equation}
\end{remark}
\begin{proof}[Proof of Theorem~\ref{t.with-fading}]
We use the propagation-loss process representation $\{Y_n\}$ 
defined in Section~\ref{ss.SINRcoverage} (which does account for
arbitrary general shadowing, but not for fading), which we enrich
by independent exponential marking $F_n$ representing Rayleigh fading. 
By Lemma~\ref{l.Yn} $\{(Y_n,F_n):n\ge1\}$ is  independently marked Poisson
point process of intensity $\Lambda(\cdot)$. Using this representation 
we can express the coverage probability of~(\ref{e.Pc_fading}) as follows
\begin{align}
\tilde P_c(T)&= \int_0^{\infty} \Prob\Bigl\{\,F_1 \geq
sT(W+I_{(s,\infty)}) \,\Bigr\}
f_{Y_1}(s)\,ds\,\nonumber\\
&=\int_0^{\infty} \calL_W(sT)\calL_{(s,\infty)}(sT)
f_{Y_1}(s)\,ds\,,\label{e.tPc-intermediate}
\end{align}
where $f_{Y_1}(s)$ is the probability density of $Y_1$, known to be
(due to Poissonianity of $\{Y_n\}$) 
$f_{Y_1}(ds)=-\frac{d}{ds}e^{-\Lambda(s)}=2a/\beta s^{2/\beta-1}$
and where $I_{(s,\infty)}$ is the random variable representing 
conditional interference (accounting for shadowing and fading) given
$Y_1=s$. Again, it is well known that $I_{(s,\infty)}$
is equal in distribution to $\sum_{Y_n>s}F_n/Y_n$ and has 
the Laplace transform which can be explicitly evaluated  as follows
\begin{align}
 \Lap_{I_{(s,\infty)}}(\xi) &=\exp\left(-\int_{s}^{\infty}  [1-\Lap_F(\xi/v)] 
\Lambda(dv)  \right)\\
&=\exp\left(-\frac{2a}{\beta}  \frac{ \xi s^{2/\beta}}{s} \frac{ 
{}_2F_1(1,1-2/\beta;2-2/\beta;-\xi/s }{(1-2/\beta) 
}  \right).
\end{align}
Plugging into~(\ref{e.tPc-intermediate}) and substituting $t=sa^{\beta/2}$ completes the proof.
\end{proof}

\begin{figure}[t!]
\begin{center}
\centerline{\includegraphics[width=1\linewidth]{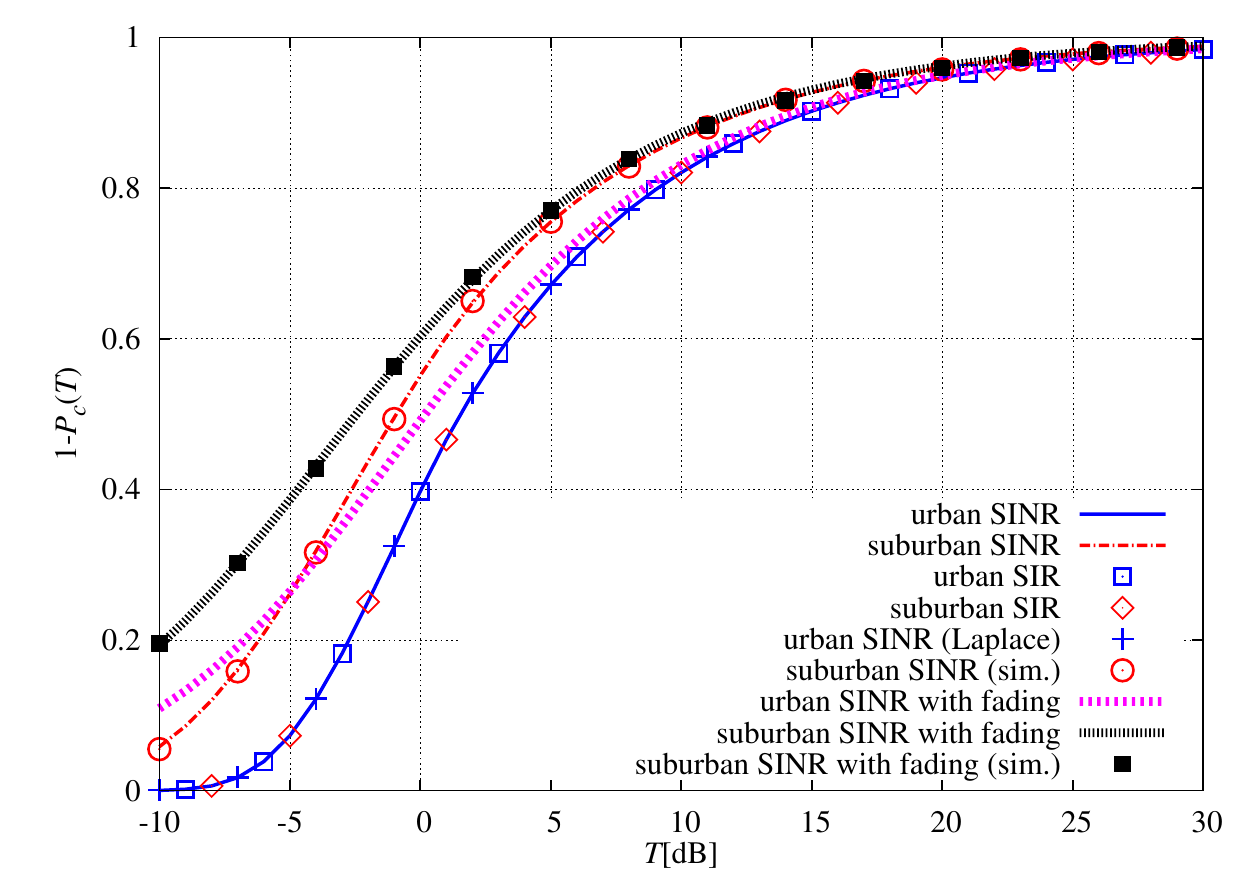}}
\vspace{-1ex}
\caption{Distribution function of SINR  from the strongest base
  station with and without fading validated by  the Laplace inversion method and simulation.
\label{f.SINR}}
\end{center}
\vspace{-3ex}
\end{figure}

\section{Numerical illustrations}
We use MATLAB implementation~\cite{paul_matlab} for all our calculations. 
We set $\beta=3.8$ and $K=6910$~km$^{-1}$ (which corresponds to the
COST Walfisch-Ikegami model for urban environment). 
The shadowing is modeled by a log-normal random
variable of expectation $1$ and  logarithmic standard deviation
$10\text{dB}$ (cf~\cite{hextopoi}) which makes $\E(S^{2/\beta})=0.516$. 
We assume noise power $-96\text{dBm}$ normalized by the base
station power  $62.2\text{dBm}$ which makes $W=10^{-15.82}.$ 
We consider two values for  the density of base stations:
$\lambda=4.619\text{km}^{-2}$, which corresponds to a ``urban''
network deployment and $\lambda=0.144\text{km}^{-2}$ for a ``suburban'' one.  
Figure~\ref{f.SINR} shows the distribution function of SINR 
from the strongest base station for both scenarios.
We validate our approach
by showing that the obtained results coincide with those of simulation
and a Laplace inversion method developed in~\cite{hextopoi}, with the latter
approach being less numerically stable and much more time-consuming.
We also plot the distribution of SIR in both scenarios (i.e. assuming
$W=0$). Both SIR curves coincide with that of the SINR in urban area, thus 
showing that for urban density of stations the network is
interference-limited, while for suburban density the impact of noise
is non-negligible. Finally, we provide curves regarding the model 
with fading (i.e. $1-\tilde P_c(T)$). We observe that the impact of
fading is non-negligible in both the urban and suburban scenario,
and stochastically decreases SINR (the respective distribution
functions are larger).

\section{Conclusion}
Cellular network models based on the Poisson
point process allow for analytic expression for many important 
characteristics. Complementing previous studies, 
in this paper we give tractable, integral expressions 
(without any Laplace transform inversion) for the distribution of
the SINR experienced by a typical user in the down-link channel 
from the $k$-th strongest base station in a single-tier cellular network.
Our signal propagation-loss model comprises of a power-law path-loss function with
arbitrarily distributed shadowing, with and without Rayleigh
fading.

\appendix


\subsection{Remaining proof of Theorem \ref{mainResult}}\label{mainResultProofPart2}
We introduce a change of variables
inspired by the $n$-dimensional spherical coordinates (for example, see
\cite[eq. (1.3)]{MUSTARD1964})
\begin{align*}
 s_1&:=u[\sin\theta_1 \sin\theta_2\dots\sin\theta_{n-1}]^{2/\beta}\\
s_2&:=u[\cos\theta_1 \sin\theta_2\dots\sin\theta_{n-1}]^{2/\beta}\\
s_3&:=u[\cos\theta_2 \sin\theta_3\dots\sin\theta_{n-1}]^{2/\beta}\\[-1ex]
&\cdots\\[-1.5ex]
s_n&:=u[\cos\theta_{n-1}]^{2/\beta}.
\end{align*}
Observe that $\sum_{i=1}^{n} s_i^{\beta}=u^{\beta} $ and 
$\prod_{i=1}^{n} s_i= 
u^n \left[\prod_{i=1}^{n}q_i \right]^{2/\beta} $,
where $q_i=q_i(\theta_i,\dots,\theta_{n-1}):=(s_i/u)^{\beta/2}$.
When $\beta=2$ our system of coordinates boils down to  
the  regular $n$-dimensional spherical coordinates, whose Jacobian is
$ \bar{J}(u,\theta_1,\dots,\theta_{n-1})
=u^{n-1} \prod_{i=1}^{n-1} \sin^{i-1}\theta_{i}$;
cf \cite[eq. (1.5)]{MUSTARD1964}).
By induction (or determinant properties and the chain rule) our coordinate
system has the corresponding Jacobian
\begin{align}
& J(u,\theta_1,\dots,\theta_{n-1})\\
&=\left(\frac{2}{\beta}\right)^{n-1}\bar{J}(u,
\theta_1,\dots,\theta_n)  \left[\prod_{i=1}^{n-1}\sin^{i}\theta_i\cos\theta_i
\right]^{2/\beta-1}\,,
\end{align}
which is clearly postive over the integration domain of interest.
Denote  $z:=W(a\Gamma(1-2/\beta))^{-\beta/2}$.
The integral in (\ref{intgralHn1}) becomes
\begin{align*}
&\int_0^{\infty}\!\!\!\!\!\int\limits_{[0,\pi/2]^{n-1}}\frac{
u^{n(\beta+1)} \left[\prod\limits_{i=1}^{n-1}\sin^{i}\theta_i\cos\theta_i
\right]^{2(\beta+1)/\beta} 
e^{-u^2} e^{-zu^{\beta}} }{u^{n\beta}
\prod\limits_{i=2}^n [q_i^2 +T_n 
] }  \\ & \times J(u,\theta_1,\dots,\theta_n) \,du\, d\theta_1\dots
d\theta_{n-1} \\
&=\left(\frac{2}{\beta}\right)^{n-1}\int_0^{\infty} u^{2n-1}  
e^{-u^2} e^{-zu^{\beta}} du \\ 
& \times\int\limits_{[0,\pi/2]^{n-1}}\frac{\prod\limits_{i=1}^{n-1} 
  \left[ \sin^{i}\theta_i\cos\theta_i \right]^{4/\beta+1} 
 [\sin\theta_{i}]^{i-1} }{\prod\limits_{i\neq j}[q_i^2 +T_n 
] }   d\theta_1\dots d\theta_{n-1} .
\end{align*}
The substitution $v_i=\sin^2 \theta_i$ makes the second integral (over
the hypercube) equal to $2^{1-n}\calJ_{n,\beta}(T_n)$, which, after defining $\eta_i$, completes the
proof in  view of~(\ref{intgralHn1}) and~(\ref{Tn}).

\pdfbookmark[0]{References}{References} 
{\small 
\bibliographystyle{IEEEtran}

}
\end{document}